\documentclass[aps,pra,showpacs,amssymb,superscriptaddress,nofootinbib,twocolumn]{revtex4-1}

\usepackage{amsmath,amsfonts,amsthm,amssymb}
\usepackage{braket}
\usepackage{setspace}
\usepackage{amsmath}
\usepackage{appendix}
\usepackage[ansinew]{inputenc}
\usepackage{bbm}
\usepackage{bm}
\usepackage{amsbsy}
\usepackage{dsfont} % for symbols like the identity: \mathds{1}
\usepackage{graphicx} % for graphics
\usepackage{epsfig}
\usepackage{epstopdf}
\usepackage{dsfont}
\usepackage{color}
\usepackage[bookmarks=true, colorlinks]{hyperref}
\makeatletter
\newcommand\org@hypertarget{}
\let\org@hypertarget\hypertarget
\renewcommand\hypertarget[2]{%
  \Hy@raisedlink{\org@hypertarget{#1}{}}#2%
  }
\makeatother

\usepackage[figure,table]{hypcap}
\usepackage{enumerate}%allows different styles of enumerate environment
\usepackage[resetlabels]{multibib}

\hypersetup{
	bookmarksnumbered,
	pdfstartview={FitH},
	citecolor={darkgreen},
	linkcolor={darkred},
	urlcolor={darkblue},
	pdfpagemode={UseOutlines}}
\definecolor{darkgreen}{RGB}{50,190,50}
\definecolor{darkblue}{RGB}{0,0,190}
\definecolor{darkred}{RGB}{238,0,0}
\usepackage{soul}
\newcommand{\nl}{\ensuremath{\hspace*{-0.5pt}}}
\newcommand{\nr}{\ensuremath{\hspace*{0.5pt}}}

\newcommand{\comm}[2]{\ensuremath{\left[\right.\! #1 \,, #2 \!\left.\right]}}

\newcommand{\expval}[1]{\ensuremath{\left\langle\right. #1 \left.\right\rangle}}

\newcommand{\pr}{^{\prime}}

\newcommand{\subtiny}[3]{\ensuremath{_{\hspace{#1 pt}\protect\raisebox{#2 pt}{\tiny{$ #3$}}}}}

\newcommand{\dv}{\ensuremath{|\hspace*{-1.3pt}|}}
\DeclareMathOperator{\artanh}{artanh}

\DeclareMathOperator{\diag}{diag}

\newcommand{\tr}{\textnormal{Tr}}
\newcommand{\djj}{d\kern-0.4em\char"16\kern-0.1em}
\newtheorem{theorem}{Theorem}%[section]
\newtheorem{lemma}{Lemma}%[section]
\newtheorem{coro}{Corollary}%[section]

\begin{document}

\title{Passivity and practical work extraction using Gaussian operations}
\author{Eric G. Brown}
\email{eric.brown@icfo.es}
\affiliation{ICFO - Institut de Ciencies Fotoniques, The Barcelona Institute of Science and Technology, 08860 Castelldefels (Barcelona), Spain}
\author{Nicolai Friis}
\email{nicolai.friis@uibk.ac.at}
\affiliation{Institute for Theoretical Physics, University of Innsbruck, Technikerstra{\ss}e 21a, 6020 Innsbruck, Austria}
\author{Marcus Huber}
\email{marcus.huber@univie.ac.at}
\affiliation{Institute for Quantum Optics and Quantum Information, Austrian Academy of Sciences, Boltzmanngasse 3, 1090 Vienna, Austria}
\affiliation{Group of Applied Physics, University of Geneva, 1211 Geneva 4, Switzerland}
\affiliation{Departament de F\'isica, Universitat Aut\`onoma de Barcelona, 08193 Bellaterra, Spain}
\affiliation{ICFO - Institut de Ciencies Fotoniques, The Barcelona Institute of Science and Technology, 08860 Castelldefels (Barcelona), Spain}

\begin{abstract}
Quantum states that can yield work in a cyclical Hamiltonian process form one of the primary resources in the context of quantum thermodynamics. Conversely, states whose average energy cannot be lowered by unitary transformations are called passive. However, while work may be extracted from non-passive states using arbitrary unitaries, the latter may be hard to realize in practice. It is therefore pertinent to consider the passivity of states under restricted classes of operations that can be feasibly implemented. Here, we ask how restrictive the class of Gaussian unitaries is for the task of work extraction. We investigate the notion of Gaussian passivity, that is, we present necessary and sufficient criteria identifying all states whose energy cannot be lowered by Gaussian unitaries. For all other states we give a prescription for the Gaussian operations that extract the maximal amount of energy. Finally, we show that the gap between passivity and Gaussian passivity is maximal, i.e., Gaussian-passive states may still have a maximal amount of energy that is extractable by arbitrary unitaries, even under entropy constraints.
\end{abstract}

\maketitle

\section{Introduction}

At the very core of quantum thermodynamics \textemdash\ which has recently seen a surge in interest from the quantum information community~\cite{GooldHuberRieraDelRioSkrzypczyk2016, MillenXuereb2016, VinjanampathyAnders2016}\textemdash\ lies the task of extracting work from quantum systems. Ideally, this is achieved by reversible cyclic processes that can be represented by unitary transformations, which, in turn, form the most basic primitive of a cyclically operating engine. However, from so-called \emph{passive} states~\cite{PuszWoronowicz1978} no work can be extracted unitarily if only a single copy of the system is available. Viewing quantum thermodynamics as a resource theory of work extraction~\cite{BrandaoHorodeckiOppenheimRenesSpekkens2013,CoeckeFritzSpekkens2016}, it is therefore tempting to view passive states as being freely available. Curiously, most passive states still contain extractable work, that can only be accessed by processing multiple copies using entangling operations \cite{AlickiFannes2013}. This has sparked interested in the role of entanglement in work extraction~\cite{HovhannisyanPerarnauLlobetHuberAcin2013} and more generally in the structure of passive states~\cite{PerarnauLlobetHovhannisyanHuberSkrzypczykTuraAcin2015}. States from which no energy can be extracted unitarily, no matter how many copies are available, are called completely passive, and without further conserved charges the unique completely passive state is the so-called thermal state. Naturally, the extraction of work from non-passive quantum states using arbitrary unitaries has therefore been the subject of many fruitful investigations (see, e.g., Refs.~\cite{AllahverdyanBalianNieuwenhuizen2004, Aberg2013, SkrzypczykShortPopescu2014, SkrzypczykSilvaBrunner2015}). This has provided useful insights into the role of coherence, correlations, and entanglement for work extraction~\cite{GemmerAnders2015, KammerlanderAnders2016, KorzekwaLostaglioOppenheimJennings2016, PerarnauLlobetHovhannisyanHuberSkrzypczykBrunnerAcin2015}, and conversely, about the work cost for establishing correlations~\cite{HuberPerarnauHovhannisyanSkrzypczykKloecklBrunnerAcin2015,BruschiPerarnauLlobetFriisHovhannisyanHuber2015,FriisHuberPerarnauLlobet2016} and coherence~\cite{MisraSinghBhattacharyaPati2016}.

However, the unitary operations that extract the maximal amount of work from a given non-passive state may be difficult to realize in practice. For example, the global entangling unitaries required to extract work from passive but not completely passive states are not feasible under realistic conditions, which already leads to~a discrepancy between theoretically extractable work and what is practically achievable. Consequently, the characterization of states as non-completely passive may fail to accurately represent how useful a quantum system is for thermodynamical work extraction in practice. It is hence of interest to study the \emph{ergotropy}~\cite{AllahverdyanBalianNieuwenhuizen2004}, i.e., the maximal amount of work that can be unitarily extracted in a Hamiltonian process, under the restriction to practically realizable transformations. An important example of such a constraint in continuous variable systems is the class of Gaussian unitary operations, which, although being standard operations in quantum optical systems (see, e.g., Refs.~\cite{Olivares2012, Weedbrooketal2012}), represents a significant restriction of the set of all possible transformations. This manifests, amongst other things, in the fact that Gaussian operations are not universal for quantum computation~\cite{LloydBraunstein1999}.

It is the aim of this paper to investigate this interesting dichotomy between what is possible in principle and what is practically feasible within quantum thermodynamics, focusing on Gaussian unitary transformations. We provide a full characterization of Gaussian passivity for multi-mode states, i.e., we give necessary and sufficient conditions to identify all (not necessarily Gaussian) quantum states from which no work can be extracted using Gaussian transformations. This characterization only requires knowledge of the first and second statistical moments of the state in question, independently of whether the state itself is Gaussian or not, and immediately provides a protocol for optimal Gaussian work extraction given any number of modes. Finally, we show that the gap between passivity and Gaussian passivity can be maximal if only the first and second moments of the state are known. That is, we show that the first and second moments of any Gaussian-passive state (which need not be a Gaussian state, and hence is not to be confused with a passive Gaussian state) are in principle compatible with those of~a state of maximal ergotropy, even under entropy constraints.

To examine the usefulness of Gaussian transformations in the context of thermodynamic work extraction, we first review the notion of passivity in Section~\ref{sec:passive states}. In Section~\ref{sec:Gaussian passive states} we introduce the notion of Gaussian passivity and formulate our main result, the characterization of all Gaussian-passive states, before examining the ergotropic gap in Section~\ref{sec:gap between passivity and Gaussian passivity}. In Section~\ref{sec:discussion} we finally discuss the consequences and applications of our results.

\section{Passive states}\label{sec:passive states}

On an elementary level, controlled engines perform their tasks based on repeated cycles, during which the system dynamics are typically changed through an external control. The resulting time evolution of the system is hence governed by~a time-dependent Hamiltonian $H(t)$. At the end of each cycle of duration $\tau$ the system is returned to its initial Hamiltonian, i.e., $H(n\tau)=H(0)\equiv H_{0}$ for any integer $n$, leading to unitary dynamics perturbed only by the necessary interactions with the environment. In this sense the unitary orbits of the input quantum states determine the fundamental limits of operation of cyclic machines, which is one of the reasons that unitary work extraction has recently attracted attention~\cite{HovhannisyanPerarnauLlobetHuberAcin2013, AlickiFannes2013, PerarnauLlobetHovhannisyanHuberSkrzypczykBrunnerAcin2015, PerarnauLlobetHovhannisyanHuberSkrzypczykTuraAcin2015, GelbwaserKlimovskyKurizki2015, LevyDiosiKosloff2016}.

Within this paradigm, the elementary processes that we consider here can in principle generate work if the average energy of a given system can be reduced by unitary operations. That is, if for a~system with Hilbert space~$\mathcal{H}$ described by~a density operator~$\rho\in L(\mathcal{H})$ there exists~a unitary operator $U\in L(\mathcal{H})$, such that
\begin{align}
    \tilde{E}   &=\,\tr\bigl(H_{0}\hspace*{1pt}U\rho\hspace*{1pt} U^{\dagger}\bigr)\,<\,\tr\bigl(H_{0}\hspace*{1pt}\rho \bigr)\,=\,E\,.
\end{align}		
States for which the average energy cannot be reduced by unitaries are called \emph{passive}. All passive states $\rho_{\mathrm{pass}}$ are diagonal in the eigenbasis $\{\ket{n}\}$ of $H_{0}$ with probability weights decreasing (not necessarily strictly) with increasing energy~\cite{PuszWoronowicz1978}, that is, passive states $\rho_{\mathrm{pass}}$ can be written as
\begin{align}
    \rho_{\mathrm{pass}}   &=\,\sum\limits_{n=0}^{d-1}p_{n}\,\ket{n}\!\!\bra{n}\,,
    \label{eq:passive states}
\end{align}
where $p_{n}\leq p_{m}$ when $E_{n}\geq E_{m}$, and $d=\dim(\mathcal{H})$, and, as usual, $0\leq p_{n}\leq1$ and $\sum_{n}p_{n}=1$, while the energy eigenstates satisfy
\begin{align}
    H_{0}\,\ket{n} &=\,E_{n}\,\ket{n}\,.
\end{align}
To see this, simply consider a two-dimensional subspace spanned by $\ket{m}$ and $\ket{n}$ with $E_{m}<E_{n}$. To decrease the average energy using a unitary on this subspace, the corresponding probability weights must satisfy $p_{m}<p_{n}$. Any state for which this isn't the case for any two-dimensional subspace is passive, and is of the form of Eq.~(\ref{eq:passive states}).

An example for~a passive state in continuous variable systems is~a product state of two thermal states of two bosonic modes at the same frequency~$\omega$ and temperature $T=1/\beta$ (in units where $\hbar=k\subtiny{0}{0}{\mathrm{B}}=1$). A bosonic mode is represented by annihilation and creation operators $a$ and $a^{\dagger}$, satisfying $\comm{a}{a^{\dagger}}=1$, and a Hamiltonian $H_{0}=\omega\,a^{\dagger}a$. The operator $a$ annihilates the vacuum state, i.e., $a\ket{0}=0$, and the eigenstates of the Hamiltonian, the Fock states, are obtained by applying the creation operators $\ket{n}=(a^{\dagger})^{n}/\sqrt{n!}\ket{0}$. A thermal state $\tau(\beta)$ of temperature $T=1/\beta$ is given by
\begin{align}
    \tau(\beta) &=\,\frac{e^{-\beta H_{0}}}{\mathcal{Z}}\,,
    \label{eq:thermal state}
\end{align}
where $\mathcal{Z}=\tr(e^{-\beta H_{0}})$ is the partition function. In the Fock basis the thermal state reads
\begin{align}
    \tau(\beta) &=\,(1-e^{-\beta\omega})\sum_{n}\,e^{-n\beta\omega}\,\ket{n}\!\!\bra{n}\,.
    \label{eq:single mode thermal states}
\end{align}
Since thermal states are the only completely passive states, any resource state for a~cyclic engine must be out of thermal equilibrium. The most elementary engine is hence a heat engine, which only needs access to two thermal baths at different temperatures. This is arguably the simplest out-of-equilibrium resource: Given two thermal baths at equal temperature, increasing the temperature of one of them can be achieved by increasing the energy of one of the systems without requiring any knowledge of its microstates.

At this point, it seems prudent to restate the above observation about these elementary heat engines in a more technical manner by reminding the reader of the elementary fact that, for two bosonic modes at the same frequency, the product state of two thermal states with different temperatures is not passive. Since we will refer to it later in the manuscript, it is instructive to examine one potential strategy to prove this statement. Consider two bosonic modes at the same frequency $\omega$, with temperatures $T_{a}$ and $T_{b}>T_{a}$, respectively. The product state of the two thermal states is
\begin{align}
    \tau(T_{a},T_{b})   &=\,
    (1-e^{-\omega/T_{a}})(1-e^{-\omega/T_{b}})\nonumber\\
    &\times\sum_{m,n}\,e^{-\omega(\tfrac{m}{T_{a}}+\tfrac{n}{T_{b}})}\,\ket{m}\!\bra{m}\otimes\ket{n}\!\bra{n}.
    \label{eq:two mode thermal state}
\end{align}
Up to a common prefactor, the matrix elements are
\begin{align}
    e^{-\omega(\tfrac{m}{T_{a}}+\tfrac{n}{T_{b}})}  &=\,
    e^{-\tfrac{\omega}{T_{a}T_{b}}(m T_{b}+n T_{a})}\,.
\end{align}
The state is not passive if there exist pairs of non-negative integers $m$, $n$ and $m\pr$, $n\pr$, such that
\begin{align}
    m \,T_{b}\,+\,n\, T_{a} &>\,m\pr\, T_{b}\,+\,n\pr\, T_{a}\,,
\end{align}
while $m\pr+n\pr>m+n\,$.
%\begin{align}
%    m\pr\,+\,n\pr  &>\,m\,+\,n\,.
%\end{align}
Now consider, e.g., the case where $m=n=x/2$, while $m\pr=0$ and $n\pr=x+1$ for some even non-negative integer $x$. In this case, the second inequality is trivially true for all $x$ and the first condition is
\begin{align}
    x(T_{a}\,+\,T_{b})/2 &>\, (x\,+\,1)\,T_{a}\,,
\end{align}
which implies $x>2T_{a}/(T_{b}-T_{a})>0$. So by picking $x$ large enough, one can always find a two-dimensional subspace, in which a unitary can reduce the average energy, proving that the state $\tau(T_{a},T_{b})$ of Eq.~(\ref{eq:two mode thermal state}) is not passive.

However, practically realizing these unitaries on arbitrary (two-dimensional) subspaces of the overall Hilbert space may prove to be practically unachievable. Therefore, it is of interest to investigate the limitations of work extraction by operations that can be feasibly implemented.

\section{Gaussian-passive states}\label{sec:Gaussian passive states}

While the fact that $\tau(T_{a},T_{b})$ of Eq.~(\ref{eq:two mode thermal state}) is not passive in principle allows the construction of a heat engine the necessary unitaries to extract energy may be difficult to realize in practice. A set of operations that are in general easier to implement are Gaussian unitaries, which are generated by Hamiltonians that are at most quadratic in the mode operators, and transform Gaussian states to Gaussian states. These, in turn, are states whose characteristic Wigner function is Gaussian (see, e.g., Ref.~\cite{Weedbrooketal2012} for a detailed review). Such states are fully described by their first and second statistical moments, that is, the expectation values of linear and quadratic combinations of the quadrature operators $\mathds{X}_{i}$, where $\mathds{X}_{2n-1}=\bigl(a_{n}+a_{n}^{\dagger}\bigr)/\sqrt{2}$ and $\mathds{X}_{2n}=-i\bigl(a_{n}-a_{n}^{\dagger}\bigr)/\sqrt{2}$, and $n=1,2,\ldots,N$ labels the $N$ modes of the system in question. The second moments are collected in the real, symmetric, and positive definite $2N\times2N$ covariance matrix $\Gamma$, with components
\begin{align}
    \Gamma_{ij} &=\,\expval{\mathds{X}_{i}\mathds{X}_{j}+\mathds{X}_{j}\mathds{X}_{i}}\,-\,2\expval{\mathds{X}_{i}}\expval{\mathds{X}_{j}}\,,
\end{align}
where $\expval{\!A\!}=\tr(A\hspace*{1pt}\rho)$ is the expectation value of the operator $A$ in the state $\rho$. For example, the thermal state of a single bosonic mode that we have encountered in Eq.~(\ref{eq:single mode thermal states}) belongs to the class of Gaussian states and is of particular interest for us here. Its first moments vanish, $\expval{\nl\mathds{X}_{i}\nl}=0$, while the covariance matrix is given by $\Gamma_{\!\mathrm{thermal}}=\coth\bigl(\beta\omega/2\bigr)\mathds{1}$.

Gaussian unitaries are described by affine maps $(S,\xi):\mathds{X}\mapsto S\mathds{X}+\xi$, where $\xi\in\mathbb{R}^{2N}$ represent phase space displacements, and $S$ are real, symplectic matrices. By definition, a symplectic operation $S$ leaves invariant the symplectic form $\Omega$ with components  $\Omega_{mn}=i\comm{\mathds{X}_{m}}{\mathds{X}_{n}}=\delta_{m,n-1}-\delta_{n,m+1}$, i.e.,
\begin{align}
    S\,\Omega\,S^{T}    &=\,\Omega\,.
\end{align}
Displacements, represented by the unitary Weyl operators $D(\xi)=\exp\bigl(i\sqrt{2}\mathds{X}^{T}\Omega\xi\bigr)$, do not affect the covariance matrix but rather shift the first moments. While all of the mentioned transformations preserve the Gaussian character of Gaussian states, one can of course also consider the effects of Gaussian transformations on any arbitrary state via the effect on the corresponding covariance matrix and vector of first moments. We are now interested in determining for which (not necessarily Gaussian) states of two noninteracting bosonic modes with frequencies $\omega_{a}$ and $\omega_{b}$ (w.l.o.g. we assume $\omega_{b}\geq\omega_{a}$), energy can be extracted using only Gaussian operations. In analogy to the previous terminology we call states \emph{Gaussian-passive} if their average energy cannot be reduced by Gaussian unitaries. The first important step in analyzing this property is the ability to identify Gaussian-passive states, which is established by the following theorem.

\begin{theorem}\label{theorem:Gaussian passive states}
Any (not necessarily Gaussian) state of two noninteracting bosonic modes with frequencies $\omega_{a}$ and $\omega_{b}\geq\omega_{a}$
%and with covariance matrix $\Gamma$ and vector of first moments $\expval{\mathbb{X}}$
is Gaussian-passive if and only if its first moments vanish, $\expval{\mathbb{X}}=0$, and its covariance matrix $\Gamma$ is either
\begin{enumerate}[(i)]
\item{\label{Theorem I cond i} in Williamson normal form~\emph{\cite{Williamson1936}}
\begin{align*}
    \Gamma  &=\,\diag\{\nu_{a},\nu_{a},\nu_{b},\nu_{b}\}
\end{align*}
with $\nu_{a}\geq\nu_{b}$ for $\omega_{a}<\omega_{b}$. Or, in the case of equal frequencies $\omega_a=\omega_b$, the state may also be}
\item{\label{Theorem I cond ii} in standard form~\emph{\cite{DuanGiedkeCiracZoller2000,Simon2000}}
\begin{align*}
    \Gamma  &=\,\begin{pmatrix}a\mathds{1}  &   C   \\  C  &   b\mathds{1}   \end{pmatrix}
\end{align*}
with $C=c\mathds{1}$, if $\omega_{a}=\omega_{b}$.}
\end{enumerate}
\end{theorem}
\begin{proof}
To prove Theorem~\ref{theorem:Gaussian passive states}, we will proceed in the following way. We will start with the most general combination of first moments $\expval{\mathbb{X}}$ and second moments $\Gamma$ that any initial state may have, before successively applying Gaussian operations (steps \hyperlink{proof displacements}{P1}-\hyperlink{proof BS}{P4}) to reduce the average energy. When we reach a state whose energy cannot be further reduced by Gaussian unitaries, we compare its energy to that of the initial state and identify under which conditions the energy has been lowered with respect to the initial state. These conditions will finally result in the identification of the characteristics of Gaussian-passive states as stated in clauses~(\ref{Theorem I cond i}) and~(\ref{Theorem I cond ii}) above.
    \begin{enumerate}[(P1)]
    \item{\hypertarget{proof displacements}
    \textbf{Displacements}:\
    As we consider noninteracting bosonic modes, the average energy of a two-mode state is given by the sum of the average energies of the individual modes. For a single mode with frequency $\omega_{a}$ and ladder operators $a$ and $a^{\dagger}$, a state described by the density operator $\rho$ has the average energy $E(\rho)=\omega_{a}\tr(\rho\hspace*{0.5mm}a^{\dagger}a)$, which can be written in terms of the state's covariance matrix $\Gamma_{a}$ and vector of first moments $\expval{\mathbb{X}_{a}}$ as
    \begin{align}
        \hspace*{5mm}E(\rho)  &=\,
        \omega_{a}\Bigl(\frac{1}{4}\bigl[\tr(\Gamma_{a})-2\bigr]+\frac{1}{4}\dv\!\expval{\nl\mathbb{X}_{a}\nl\!}\!\dv^{2}\Bigr)\,,
        \label{eq:energy first and second moments}
    \end{align}
    where $\dv \cdot\dv$ is the (Euclidean) norm. Since displacements change the first moments but leave the second moments invariant, the energy of the state can always be decreased by shifting $\expval{\nl\mathbb{X}_{a}\nl\!}$ to the null vector. Conversely, every state with nonvanishing first moments cannot be Gaussian-passive, since its energy can be lowered by appropriate displacements. From here on we may hence consider only states for which the energy has been reduced by displacements as much as possible, such that for each mode one has $\expval{\nl\mathbb{X}_{a}\nl\!}=0$. In the following, one may then apply Gaussian unitaries represented by symplectic transformations $S$, which leave the zero first moments invariant.}
    \item{\textbf{Local symplectic operations}:\
    In the next step, we note that every two-mode covariance matrix $\Gamma$ can be brought to the standard form~$\Gamma_{\mathrm{st}}$~\cite{DuanGiedkeCiracZoller2000,Simon2000} by local symplectic operations $S_{\mathrm{loc}}=S_{\mathrm{loc},a}\oplus S_{\mathrm{loc},b}$, that is,
        \begin{align}
            S_{\mathrm{loc}}\,\Gamma\,S_{\mathrm{loc}}^{T}  &=\,\Gamma_{\mathrm{st}}\,=\,
            \begin{pmatrix}
                a\,\mathds{1}   &   C   \\  C   &   b\,\mathds{1}
            \end{pmatrix}\,,
            \label{eq:cov matrix standard form}
        \end{align}
        where $C=\diag\{c_{1},c_{2}\}$. Each of the single-mode symplectic operations $S_{\mathrm{loc},i}$ $(i=a,b)$ can be decomposed into phase rotations and single-mode squeezing as
    \begin{align}
        S_{\mathrm{loc},i}    &=\,R(\theta_{i})\,S(r_{i})\,R(\phi_{i})\,.
    \end{align}
    For some real angles $\theta_{i}$ and $\phi_{i}$, and real squeezing parameters $r_{i}$, these local operations take the form
    \begin{align}
        \ \ \ \ \ R(\theta_{i})=\begin{pmatrix}   \cos\theta_{i}  &   \sin\theta_{i}  \\  -\sin\theta_{i} &   \cos\theta_{i}  \end{pmatrix}\,,\ \
        S(r_{i})    \,=\,\begin{pmatrix}   e^{-r_{i}}  &   0  \\  0 &   e^{r_{i}}  \end{pmatrix}\,.
    \end{align}
    Conversely, this means that we can write the covariance matrix $\Gamma$ as
    \begin{align}
        \Gamma  &=\,(S_{\mathrm{loc}}^{-1})\,\Gamma_{\mathrm{st}}\,(S_{\mathrm{loc}}^{-1})^{T}\,,
    \end{align}
    where the inverse operations are also local symplectic transformations $S_{\mathrm{loc}}^{-1}=S_{\mathrm{loc},a}^{-1}\oplus S_{\mathrm{loc},b}^{-1}$. The single-mode inverses $S_{\mathrm{loc},i}^{-1}$ are simply
    \begin{align}
        S_{\mathrm{loc},i}^{-1} &=\,R(-\phi_{i})\,S(-r_{i})\,R(-\theta_{i})\,.
    \end{align}
    This allows us to express the energy of the state described by the covariance matrix $\Gamma$ as
    \begin{align}
        \ \ \ \ \ E(\Gamma)   &=\,\frac{\omega_{a}}{2}\bigl(a\cosh(2r_{a})-1\bigr)+\frac{\omega_{b}}{2}\bigl(b\cosh(2r_{b})-1\bigr)\,.
    \end{align}
    Since $\cosh(2r_{i})\geq1$, it becomes clear that the energy of a state with covariance matrix $\Gamma$ can be lowered by local symplectic operations until $\Gamma$ reaches the standard form. Consequently, states for which $\Gamma\neq\Gamma_{\mathrm{st}}$ are not Gaussian-passive, whereas states with covariance matrices in the standard form may still have energy that can be reduced by global symplectic transformations.}
    \item{\hypertarget{proof TMS}\textbf{Two-mode squeezing}:\ After using local Gaussian operations to extract as much energy as possible, one is hence left with a state whose covariance matrix is in the standard form of Eq.~(\ref{eq:cov matrix standard form}). The local covariance matrices of each mode are then proportional to the identity, $a\mathds{1}$ and $b\mathds{1}$, but the off-diagonal block $C$ may have two different diagonal elements $c_{1}$ and $c_{2}$. If this is the case, we can apply a two-mode squeezing operation to reduce the energy and bring the covariance matrix to a form in which the off-diagonal block is proportional to the identity as well. The symplectic representation of this global transformation is
        \begin{align}
            S_{\mathrm{TMS}}    &=\,\begin{pmatrix}
            \cosh(r)\mathds{1}  &   \sinh(r)\sigma_{z}\\
            \sinh(r)\sigma_{z}  &   \cosh(r)\mathds{1}
            \end{pmatrix}\,,
        \end{align}
        where $\sigma_{z}=\diag\{1,-1\}$ is the usual Pauli matrix and the squeezing parameter $r$ that achieves equal off-diagonal elements is given by
        \begin{align}
            r   &=\,-\tfrac{1}{2}\artanh\Bigl(\frac{c_{1}-c_{2}}{a+b}\Bigr)\,.
            \label{eq:TMS parameter minimum}
        \end{align}
        To show that this transformation always reduces the average energy, we compute the energy $E(\tilde{\Gamma})$ associated to the two-mode squeezed covariance matrix $\tilde{\Gamma}=S_{\mathrm{TMS}}\Gamma_{\mathrm{st}}S_{\mathrm{TMS}}^{T}$ and find
        \begin{align}
            \hspace*{7mm}E(\tilde{\Gamma})   &=\frac{\omega_{a}}{2}\bigl[a\cosh^{2}(r)+b\sinh^{2}(r)\bigr]\nonumber\\
            &\ +\frac{\omega_{b}}{2}\bigl[b\cosh^{2}(r)+a\sinh^{2}(r)\bigr]\nonumber\\
            &\ +\frac{\omega_{a}+\omega_{b}}{4}\bigl[(c_{1}-c_{2})\sinh(2r)-1\bigr]\,.
        \end{align}
        We then take the derivative with respect to $r$ and set $\partial E(\tilde{\Gamma})/\partial r=0$, which provides the condition
        \begin{align}
            \hspace*{7mm}(a+b)\sinh(2r)+(c_{1}-c_{2})\cosh(2r)   &=\,0\,,
        \end{align}
        which in turn is solved by $r$ from Eq.~(\ref{eq:TMS parameter minimum}). It is then easy to check that for this value of $r$ we have $\partial^{2} E(\tilde{\Gamma})/\partial r^{2}>0$, indicating that the energy is minimal for the specified value of the squeezing parameter. The two-mode squeezing transformation with this strength hence reduces the energy. While the off-diagonal block of the covariance matrix is proportional to the identity after this operation, the local covariance matrices are generally not of this form, albeit still being diagonal. We can then use local rotations $R(\vartheta,\varphi)=R(\vartheta)\oplus R(\varphi)$, which leave the energy invariant, to bring the covariance matrix back to the standard form where every $2\times2$ subblock is now proportional to the identity, i.e.,
        \begin{align}
            \widehat{\Gamma}\,=\,R(\vartheta,\varphi)\,\tilde{\Gamma}\,R^{T}\!(\vartheta,\varphi)  &=\,\begin{pmatrix}
                \tilde{a}\,\mathds{1}   &   c\,\mathds{1}   \\  c\,\mathds{1}   &   \tilde{b}\,\mathds{1}
            \end{pmatrix}\,.
            \label{eq:all blocks prop identity cov matrix}
        \end{align}
        In some circumstances, the third step of the protocol can be seen as the conversion of Gaussian entanglement into work. Note that the previous two steps consist of local unitaries, and hence leave any entanglement measure invariant. If the initial state is a Gaussian state, the form of $\widehat{\Gamma}$ in Eq.~(\ref{eq:all blocks prop identity cov matrix}) further indicates that no more entanglement is present after step~\hyperlink{proof TMS}{P3}, since a nonnegative determinant of the $2\times2$ off-diagonal block is a sufficient separability criterion for two-mode Gaussian states~\cite{Simon2000}. For any Gaussian state, the presence of entanglement hence indicates that the energy can be lowered by Gaussian unitaries in the third step. However, the fact that the energy of a Gaussian state can be lowered in step~\hyperlink{proof TMS}{P3}, does not imply that the initial state is entangled~\cite{BruschiPerarnauLlobetFriisHovhannisyanHuber2015}. Moreover, if the initial state is not Gaussian, the final state after step~\hyperlink{proof TMS}{P3} may still be entangled in general.}
    \item{\hypertarget{proof BS}
        \textbf{Beam splitting}:\ Having reached a state with a covariance matrix as in Eq.~(\ref{eq:all blocks prop identity cov matrix}), we have exhausted all local Gaussian operations as well as two-mode squeezing to lower the energy. In particular, at this point we know that applying any local or global squeezing transformation can only increase the energy. This leaves only the beam splitting transformation as a last Gaussian unitary that we still have at our disposal. This transformation, represented by the global orthogonal symplectic matrix
        \begin{align}
            S_{\mathrm{BS}}(\theta) &=\,\begin{pmatrix} \cos(\theta)\,\mathds{1}  &   \sin(\theta)\,\mathds{1} \\ \sin(\theta)\,\mathds{1}  &   -\cos(\theta)\,\mathds{1}\end{pmatrix}
        \end{align}
        for real values of $\theta$, is an optically passive transformation. That is, it leaves the average excitation number unchanged. If the frequencies of the two modes are the same, $\omega_{a}=\omega_{b}$, then such a transformation obviously also leaves the average energy unchanged. In this case, the energy of the state cannot be further lowered by any Gaussian unitary and we conclude that the state is hence Gaussian-passive, which proves clause~(\ref{Theorem I cond ii}) of Theorem~\ref{theorem:Gaussian passive states}.

        If the frequencies are not the same we may assume w.l.o.g. that $\omega_{a}<\omega_{b}$. Then, the energy can be lowered by shifting as many excitations as possible to the lower frequency mode. To prove this rigorously, we compute the average energy of
        \begin{align}
            \Gamma_{\mathrm{GP}}    &=\,S_{\mathrm{BS}}(\theta)\,\widehat{\Gamma}\,S_{\mathrm{BS}}^{T}(\theta)\,,
        \end{align}
        for which we find
        \begin{align}
            &\ \hspace*{7mm} E(\Gamma_{\mathrm{GP}}) =\frac{\omega_{a}}{2}\bigl[a\cos^{2}\!(\theta)+b\sin^{2}\!(\theta)+c\sin(2\theta)-1\bigr]\nonumber\\
            &\ \hspace*{7mm} +\frac{\omega_{b}}{2}\bigl[b\cos^{2}\!(\theta)+a\sin^{2}\!(\theta)-\sin(2\theta)-1\bigr].
        \end{align}
        Similarly as for the two-mode squeezing we then set $\partial E(\Gamma_{\mathrm{GP}})\partial\theta=0$ and find that the energy is minimized when
        \begin{align}
            \theta  &=\,\begin{cases}
            \tfrac{1}{2}\arctan(\tfrac{2c}{a-b}) &   \mbox{if}\ a\geq b\\[1mm]
            \tfrac{1}{2}\arctan(\tfrac{2c}{a-b})+\tfrac{\pi}{2} &   \mbox{if}\ a<b
            \end{cases}\,.
            \label{eq:beam splitter angle}
        \end{align}
        The resulting covariance matrix $\Gamma_{\mathrm{GP}}=\diag\{\nu_{a},\nu_{a},\nu_{b},\nu_{b}\}$ is in Williamson normal form~\cite{Williamson1936}, its eigenvalues coincide with its symplectic eigenvalues, and the lower frequency mode now has the higher population, $\nu_{a}\geq\nu_{b}$. Any further Gaussian unitary applied to this final state would bring the covariance matrix (and/or the first moments) to a form that would allow reducing the energy via one (or several) of the steps~\hyperlink{proof displacements}{P1}-\hyperlink{proof BS}{P4}. The corresponding symplectic operations leave the symplectic eigenvalues invariant. Consequently, the second moments of the initial state uniquely determine the associated Gaussian-passive covariance matrix $\Gamma_{\mathrm{GP}}$. That is, $\Gamma_{\mathrm{GP}}$ is the only covariance matrix with symplectic spectrum $\{\nu_{a},\nu_{a},\nu_{b},\nu_{b}\}$ whose energy cannot be lowered by Gaussian unitaries. (If $\omega_{a}=\omega_{b}$, the Gaussian-passive covariance matrix is not unique, but is determined only up to arbitrary optically passive transformations.) We therefore arrive at the conclusion that the state associated to the covariance matrix $\Gamma_{\mathrm{GP}}$ is Gaussian-passive. Any state whose covariance matrix is not of this form can be subjected to one (or several) of the steps~\hyperlink{proof displacements}{P1}-\hyperlink{proof BS}{P4} to reduce its average energy, and is hence not Gaussian-passive, which concludes the proof.}
    \end{enumerate}
    \vspace*{-5mm}
\end{proof}
Note that for any given initial state (which need not be Gaussian), the corresponding Gaussian-passive state is not unique, because the operations~\hyperlink{proof displacements}{P1}-\hyperlink{proof BS}{P4} do not commute. For instance, applying the operations of step~\hyperlink{proof displacements}{P1} after any of the other steps leads to different final states that have the same first and second moments, and hence the same final energy. The symplectic eigenvalues of the initial state hence uniquely define the lowest energy that can be reached via Gaussian unitaries, but several (non-Gaussian) states (equivalent up to energy conserving Gaussian unitaries) may be compatible with the corresponding Gaussian-passive covariance matrix.

A corollary that follows immediately from Theorem~\ref{theorem:Gaussian passive states} concerns the extension to an arbitrary number of modes.
\begin{coro}
An arbitrary state of $n$ bosonic modes is Gaussian-passive if and only if all of its two-mode marginals are Gaussian-passive.
\end{coro}
\begin{proof}
To prove this statement, simply note that all Gaussian unitaries can be decomposed into sequences of operations on one or two modes. Consequently, if a state admits no two-mode marginal whose energy can be lowered by Gaussian unitaries, then the overall state must be Gaussian-passive.
\end{proof}

An interesting example for a Gaussian-passive state of two modes with different frequencies is that of a product of single-mode thermal states, in which each mode has a different temperature. In this case the symplectic eigenvalues are $\nu_{i}=\coth(\frac{\omega_{i}}{2T_{i}})$ and for $T_{b}\neq0$ the condition $\nu_{a}>\nu_{b}$ for Gaussian passivity can be expressed as
\begin{align}
    \frac{\omega_{a}}{\omega_{b}}   &<\,\frac{T_{a}}{T_{b}}\,.
    \label{eq:gaussian passive thermal state condition}
\end{align}
Now, recall from Section~\ref{sec:passive states} that we know that within the framework of general operations the product states of two thermal states at different temperatures is not passive, regardless of the frequencies of the modes involved. However, as we have seen, such a state may nonetheless be Gaussian-passive depending on the relation between the local temperatures and frequencies.

Here, a word of caution is in order. Since (Gaussian) unitaries leave the purity $\tr(\rho^{2})=1/\sqrt{\det(\Gamma)}$ unchanged, one may be tempted to (falsely) conclude that the existence of (Gaussian) states that have the same purity as a given Gaussian-passive state but a lower average energy means that one may further reduce the energy of Gaussian-passive states beyond what is stated in Theorem~\ref{theorem:Gaussian passive states}. For example, for $\nu_{a}\nu_{b}<\omega_{b}/\omega_{a}$ the Gaussian state with covariance matrix
\begin{align}
    \Gamma\pr   &=\,\begin{pmatrix} \nu_{a}\nu_{b}\,\mathds{1}  &   0   \\  0   &   \mathds{1}\end{pmatrix}\,,
\end{align}
has the same purity as the Gaussian-passive state specified in clause~(\ref{Theorem I cond i}). However, such states cannot be reached by Gaussian unitaries if their symplectic eigenvalues ($\nu_{a}\nu_{b}$ and $1$ in the example) do not match those of the original state. In general, there may not even exist a non-Gaussian unitary (even if it preservers the Gaussian character of the specific state in question) that transforms the corresponding states into each other. Finally, note that all passive states are obviously Gaussian-passive, but the converse is not true.

\section{The gap between passivity and Gaussian passivity}\label{sec:gap between passivity and Gaussian passivity}

Given the characterization of a given state as Gaussian-passive, it is now natural to ask how much extractable energy is potentially sacrificed by the restriction to Gaussian unitary orbits, rather than general unitary transformations. Suppose that one only has knowledge of and access to the first and second moments of an arbitrary state of two bosonic modes. With this information, which can practically be easily obtained in several ways (see, e.g., Ref.~\cite{RuppertUsenkoFilip2016}), one may use Gaussian unitaries to lower the energy of the state until reaching a Gaussian-passive state. One may then wonder how much more energy could have been extracted if general unitary operations could be applied. The answer to this question of course depends on the particular state in question. So far, we have only fixed the first and second moments, which identifies states uniquely only if they are Gaussian. It is hence crucial to understand which (non-Gaussian) states are in general compatible with a given set of first and second moments. A first important observation can be phrased in the following lemma.

\begin{lemma}\label{lemma:Gaussian passive pure state}
The first and second moments of any Gaussian-passive state are compatible with a (non-Gaussian) pure state for which the entire energy is extractable by unitary transformations.
\end{lemma}
\begin{proof}
To prove the lemma first note that any Gaussian-passive state of an arbitrary number of modes with different frequencies (clause~(\ref{Theorem I cond i}) of Theorem~\ref{theorem:Gaussian passive states}) has a locally thermal covariance matrix with different effective temperatures for each mode. In this case it is therefore enough to consider a single mode in a thermal state with arbitrary temperature, and show that there exists a pure state with the same first and second moments. If such a pure state exists for a single mode for any temperature, then one can certainly find pairs of states of this kind whose tensor product is compatible with a Gaussian-passive, locally thermal two-mode state.\\

In the case that the covariance matrix has nonzero off-diagonal blocks, i.e., if clause~(\ref{Theorem I cond ii}) of Theorem~\ref{theorem:Gaussian passive states} applies, the covariance matrix can be brought to the locally thermal form by an energy conserving, Gaussian unitary, that is, a beam splitting transformation with angle $\theta$ given by Eq.~(\ref{eq:beam splitter angle}). Then, as before, one is required to find a pure state that matches the resulting locally thermal covariance matrix. Applying the inverse of the beam splitting operation to this state, one finally obtains a pure two-mode state compatible also for Gaussian-passive states with non-diagonal covariance matrices.\\

To identify the pure states in question, recall that the first moments of a Gaussian-passive state must vanish. This is also the case for all Fock states $\ket{n}=(a^{\dagger})^{n}/\sqrt{n!}\ket{0}$. Indeed, this is even true for all superpositions of Fock states that differ by two or more excitations, for instance, all states of the form $\sum_{k}c_{k}\ket{n+m\,k}$ for any $n,m\in\mathds{N}_{0}$ and $m\geq2$. Restricting to this family of states we are interested in identifying those members that also have the second moments of a thermal state. This is achieved by considering states that are superpositions of Fock states that differ by three (or more) excitations ($m\geq3$), for instance, $\sqrt{p}\ket{n}+\sqrt{\raisebox{-1pt}{$1-$}p}\ket{n+3}$ for $0\leq p\leq1$. For such a state the covariance matrix takes the form of a thermal state $\Gamma=\nu\mathds{1}$, where $\nu=\bigl(2\nr p\nr\nr n+2\nr(1\nl-\nl p)(n\nl+\nl3)\nl+\nl1\bigr)$. By selecting the discrete value $n\in\mathds{N}_{0}$ and the continuous parameter $p$ appropriately, the second moments of this state can be chosen to match those of the desired Gaussian-passive state. The thus constructed state $\rho$ is clearly pure, and its (non-equilibrium) free energy $F(\rho)=E(\rho)-T\nr S(\rho)$, where $S(\rho)=-\tr\bigl[\rho\ln(\rho)\bigr]$ is the von~Neumann entropy, is hence identical to its average energy~$E(\rho)$. The latter can of course be lowered to zero by a (non-Gaussian) unitary by rotating the pure state towards the vacuum state.
\end{proof}

As we have seen in Lemma~\ref{lemma:Gaussian passive pure state}, if only the first and second moments of a state are known and the state is Gaussian-passive, in principle all (or none) of the state's energy may be extractable. In other words the gap between the free energy, i.e., the energy extractable by general unitary transformations, and the energy that can be extracted using only Gaussian unitaries is maximal. However, for such a maximal gap both the initial and final state must be pure, since we are applying only (Gaussian) unitary transformations, which leave the spectrum (and hence the entropy) unchanged. As most machines operate at an ambient temperature that is above zero and the second law implies that it is highly unlikely for any state to fall below the entropy of the corresponding thermal state, a maximal gap in the above sense may not occur in practice.

It is therefore reasonable to assume that, in addition to the first and second moments, also a lower bound on the entropy of the state is known. Given some nonzero entropy $S(\rho)=-\tr\bigl[\rho\ln(\rho)\bigr]$, the average energy of the state is bounded from below and may not be lowered arbitrarily\footnote{We implicitly assume that the spectrum of the Hamiltonian is bounded from below, i.e., a ground state exists.}. In such a case, it is of interest to ask whether the free energy gap is still maximal. That is, we ask: Does every Gaussian-passive state with entropy\footnote{Note that the entropy is not determined by the second moments alone, since the Gaussian-passive state need not be a Gaussian state.} $S_{0}$ admit a state $\rho$ that has the same entropy, $S(\rho)=S_{0}$, and the same first and second moments $\expval{\mathbb{X}}=0$ and $\Gamma$, but whose energy $E(\rho)$ [which is determined by $\expval{\mathbb{X}}$ and $\Gamma$ via Eq.~(\ref{eq:energy first and second moments})] may be lowered to the minimal value $E_{0}$ that is compatible with $S_{0}$ using (arbitrary) unitary transformations?

\begin{theorem}\label{theorem:extractable energy gap}
The first and second moments of any Gaussian-passive state with entropy $S_{0}$ are compatible with a (non-Gaussian) state of the same entropy for which the maximal amount of energy (the energy difference to the thermal state of entropy $S_{0}$) is extractable by unitary transformations.
\end{theorem}

\begin{proof}
For a Gaussian-passive state with fixed first and second moments ($\expval{\mathbb{X}}=0$ and $\Gamma$), the energy is also fixed, see Eq.~(\ref{eq:energy first and second moments}). In addition, we assume that the entropy of the initial (Gaussian-passive) state is $S_{0}$. Clearly, any previous Gaussian unitaries or possible general unitary transformations on the closed system that are yet to be carried out must leave this entropy invariant. On the other hand, the state $\rho$ that minimizes the energy $E(\rho)$ at a fixed entropy $S_{0}$ is the thermal state of Eq.~(\ref{eq:thermal state}). Since we cannot change the spectrum using unitary transformations, we hence have to show that for every Gaussian-passive state at entropy $S_{0}$ there exists a state $\rho$ that has the same spectrum as a thermal state of entropy $S_{0}$, but whose first and second moments (and hence its energy) match those of the Gaussian-passive state.

The strategy to show that this is possible is to start from the thermal state and manipulate it using unitary transformations to reach the desired first and second moments. In this way the spectrum of the state is preserved. In particular, we know that the spectrum is also invariant under the possible application of an energy-conserving beam splitting operation in the case that the covariance matrix of the Gaussian-passive state is not diagonal (clause~(\ref{Theorem I cond ii}) of Theorem~\ref{theorem:Gaussian passive states}). Consequently, we can again focus on proving the statement of Theorem~\ref{theorem:extractable energy gap} for single-mode Gaussian-passive states with thermal covariance matrices, as we have argued in the proof of Lemma~\ref{lemma:Gaussian passive pure state}. For each of these local single-mode covariance matrices the diagonal elements are identical and linear functions of the energy, see Eq.~(\ref{eq:energy first and second moments}). The first moments as well as the off-diagonals of the covariance matrix of both the thermal state and the initial state vanish. We therefore restrict to rotations in subspaces of Fock states that differ by three (or more) excitations to keep it that way.

We now just have to show that this method allows increasing the energy of the thermal state to reach the energy of any single-mode Gaussian-passive state, which also fixes the desired nonzero second moments. For any specified energy this can be achieved by continuously rotating in the subspace spanned by the states $\ket{0}$ and $\ket{n}$ for some sufficiently large $n\geq3$. Since the thermal state is (i) diagonal in the Fock basis, (ii) the eigenvalues are strictly decreasing with increasing $n$, and (iii) the Hilbert space is infinite-dimensional, one may reach arbitrarily large energies at a fixed entropy. Finally, because a Gaussian-passive state with the same first and second moments (and therefore same energy), and with the same entropy $S_{0}$ as the Gaussian-passive initial state can be reached unitarily from the minimal energy thermal state, the converse must also be true.
\end{proof}

It is quite remarkable to note that the proof of Theorem~\ref{theorem:extractable energy gap} makes use of the infinite-dimensionality of the Hilbert space, which is reminiscent of the famous Hilbert hotel paradox (see, e.g., Ref.~\cite[p.~17]{Gamow1947HilbertHotel}). The fact that the Hilbert space is infinite-dimensional is crucial to give the necessary freedom to be able to unitarily increase the energy of any thermal state to arbitrary values without introducing nonzero first moments or off-diagonal second moments. In any finite dimension this is not possible. In practice, one may encounter systems that are effectively finite-dimensional, which would place limitations on the applicability of Theorem~\ref{theorem:extractable energy gap}. This could lead to a potential reduction of the ergotropy gap.

Nonetheless, it is interesting to observe that the infinite dimensions of the Hilbert space may even allow extending the statement of Theorem~\ref{theorem:extractable energy gap} to cases where more than the first two statistical moments of the Gaussian-passive state are known. For instance, suppose an expectation value of a cubic combination of mode operators such as $\expval{a^{3}}$ was known. In this example, one could rotate in a subspace spanned by two Fock states separated by three excitations (e.g., $\ket{k}$ and $\ket{k+3}$ for some appropriate value of $k$) to arrange for the desired expectation value without changing the lower order moments, energy, or entropy.

\section{Conclusion}\label{sec:discussion}

In this article, we have investigated the fundamental thermodynamic problem of work extraction from continuous-variable quantum systems under the restriction to Gaussian unitaries. These operations can typically be easily implemented in quantum optics experiments, whereas the general unitary transformations that may be required to extract work from a given non-passive state may be extremely challenging to realize. To capture the limitations of this restricted class of operations for the task at hand we have introduced the notion of Gaussian passivity. We have given necessary and sufficient criteria for identifying Gaussian-passive states (whose energy may not be reduced by Gaussian unitaries) based on the first and second statistical moments of an arbitrary number of modes. Furthermore, we have shown that although the first two statistical moments provide complete information about the Gaussian ergotropy (the maximal amount of energy extractable in a Gaussian unitary process), the gap to the non-Gaussian ergotropy may always be maximal if the state is not fully known, even under entropy constraints.\\

This trade-off between usefulness and severe limitation of Gaussian operations comes as no surprise and is a recurring feature in continuous-variable quantum information. For instance, Gaussian operations are known not to be universal for computational tasks~\cite{LloydBraunstein1999}. Similar properties have also been described in~a quantum thermodynamical framework of converting work and correlations. There it was found that, while Gaussian operations provide optimal scaling for the creation of entanglement using large input energies, they cannot create entanglement with finite energy at arbitrary temperatures~\cite{BruschiFriisFuentesWeinfurtner2013, BruschiPerarnauLlobetFriisHovhannisyanHuber2015}.\\

While uncovering and quantifying the restrictiveness of Gaussian operations in the thermodynamic context, our results also provide practical strategies for the implementation of quantum heat engines based on Gaussian operations in quantum optical architectures. In particular, the steps~\hyperlink{proof displacements}{P1}-\hyperlink{proof BS}{P4} of the proof of Theorem~\ref{theorem:Gaussian passive states} can be viewed as a set of instructions for Gaussian work extraction.

\begin{acknowledgments}
We thank Michael Jabbour for interesting discussions on the topic and Chappy Bibound{\'e} for moral support. E.~G.~B. acknowledges support by the Natural Sciences and Engineering Research Council of Canada. E.~G.~B. and N.~F. thank the Universit{\'e} de Gen{\`e}ve for hospitality. N.~F. is grateful to the organizers of the workshop RACQIT for financial support and to the Foo Fighters for providing inspiration whilst proving Theorem~\ref{theorem:Gaussian passive states}. M.~H. acknowledges funding from the Austrian Science Fund (FWF) through the  START  project Y879-N27, the Swiss National Science Foundation (AMBIZIONE PZ00P2\_161351), from the Spanish MINECO through Project No.~FIS2013-40627-P and the Juan de la Cierva fellowship (JCI 2012-14155), from the Generalitat de Catalunya CIRIT Project No.~2014 SGR 966, and from the EU STREP-Project ``RAQUEL", as well as the EU COST Action MP1209 ``Thermodynamics in the quantum regime".
\end{acknowledgments}


\begin{thebibliography}{99}
%\bibliographystyle{srt}

%%[1]
\bibitem{GooldHuberRieraDelRioSkrzypczyk2016}
J.~Goold, M.~Huber, A.~Riera, L.~del~Rio, and P.~Skrzypczyk,\
\emph{The role of quantum information in thermodynamics \textemdash\ a topical review},\
%%DOI:\ 10.1088/1751-8113/49/14/143001
\href{http://dx.doi.org/10.1088/1751-8113/49/14/143001}{J.\ Phys.\ A:\ Math.}\
\href{http://dx.doi.org/10.1088/1751-8113/49/14/143001}{Theor.\ \textbf{49}, 143001 (2016)}\
[\href{http://arxiv.org/abs/1505.07835}{arXiv:1505.07835}].
%%IOP Style
%%Goold~J, Huber~M, Riera~A, del~Rio~L and Skrzypczyk~P\ 2016\ \textit{J.\ Phys.\ A:\ Math.\ Theor.}\ \href{http://dx.doi.org/10.1088/1751-8113/49/14/143001}{\textbf{49} 143001}\
%%\textit{Preprint}\ arXiv:\href{http://arxiv.org/abs/1505.07835}{1505.07835} [quant-ph]

%%[2]
\bibitem{MillenXuereb2016}
J.~Millen and A.~Xuereb,\
\emph{Perspective on quantum thermodynamics},\
%%title of arxiv v1: \emph{Perspective: Quantum Thermodynamics},\
\href{http://dx.doi.org/10.1088/1367-2630/18/1/011002}{New\ J.\ Phys.\ \textbf{18}, 011002 (2016)}\
[\href{http://arxiv.org/abs/1509.01086}{arXiv:1509.01086}].
%%IOP Style
%%Millen~J and Xuereb~A\ 2016\ \textit{New\ J.\ Phys.}\ \href{http://dx.doi.org/10.1088/1367-2630/18/1/011002}{\textbf{18} 011002}\
%%\textit{Preprint}\ arXiv:\href{http://arxiv.org/abs/1509.01086}{1509.01086} [quant-ph]

%%[3]
\bibitem{VinjanampathyAnders2016}
S.~Vinjanampathy and J.~Anders,\
\emph{Quantum Thermodynamics},\
%%DOI:\ 10.1080/00107514.2016.1201896
\href{http://dx.doi.org/10.1080/00107514.2016.1201896}{Contemp.\ Phys.\ \textbf{57}, 1 (2016)}\
[\href{http://arxiv.org/abs/1508.06099}{arXiv:1508.06099}].
%%IOP Style
%%Vinjanampathy~S and Anders~J\ 2016\ \textit{Contemp.\ Phys.}\ \href{http://dx.doi.org/10.1080/00107514.2016.1201896}{\textbf{57} 1\textendash35}\
%%\textit{Preprint}\ arXiv:\href{http://arxiv.org/abs/1508.06099}{1508.06099} [quant-ph]

%%[4]
\bibitem{PuszWoronowicz1978}
W.~Pusz and S.~L.~Woronowicz,\
\emph{Passive states and KMS states for general quantum systems},\
%%DOI:\ 10.1007/BF01614224
\href{http://dx.doi.org/10.1007/BF01614224}{Commun.\ Math.}\
\href{http://dx.doi.org/10.1007/BF01614224}{Phys.\ \textbf{58}, 273 (1978)}.
%%Full text open access: http://projecteuclid.org/euclid.cmp/1103901491
%%IOP Style
%%Pusz~W and Woronowicz~S~L\ 1978\ \textit{Commun.\ Math.\ Phys.}\ \href{http://dx.doi.org/10.1007/BF01614224}{\textbf{58} 273\textendash90}

%%[5]
\bibitem{CoeckeFritzSpekkens2016}
B.~Coecke, T.~Fritz, and R.~W.~Spekkens,\
\emph{A mathematical theory of resources},\
%%DOI:\ 10.1016/j.ic.2016.02.008
\href{http://dx.doi.org/10.1016/j.ic.2016.02.008}{Inform.\ Comput.\ \textbf{250}, 59 (2016)}\
[\href{http://arxiv.org/abs/1409.5531}{arXiv:1409.5531}].
%%IOP Style
%%Coecke~B, Fritz~T and Spekkens~R~W\ 2016\ \textit{Inform.\ Comput.}\ \href{http://dx.doi.org/10.1016/j.ic.2016.02.008}{\textbf{250} 59\textendash86}\
%%\textit{Preprint}\ arXiv:\href{http://arxiv.org/abs/1409.5531}{1409.5531} [quant-ph]

%%[6]
\bibitem{BrandaoHorodeckiOppenheimRenesSpekkens2013}
F.~G.~S.~L.~Brand{\~a}o, M.~Horodecki, J.~Oppenheim, J.~M.~Renes, and R.~W.~Spekkens,\
\emph{The Resource Theory of Quantum States Out of Thermal Equilibrium},\
%%DOI:\ 10.1103/PhysRevLett.111.250404
\href{http://dx.doi.org/10.1103/PhysRevLett.111.250404}{Phys.}\
\href{http://dx.doi.org/10.1103/PhysRevLett.111.250404}{Rev.\ Lett.\ \textbf{111}, 250404 (2013)}\
[\href{http://arxiv.org/abs/1111.3882}{arXiv:1111.3882}].
%%IOP Style
%%Brand{\~a}o~F~G~S~L, Horodecki~M, Oppenheim~J, Renes~J~M, Spekkens~R~W\ 2013\ \textit{Phys.\ Rev.\ Lett.}\
%%\href{http://dx.doi.org/10.1103/PhysRevLett.111.250404}{\textbf{111} 250404}\
%%\textit{Preprint}\ arXiv:\href{http://arxiv.org/abs/1111.3882}{1111.3882} [quant-ph].

%%[7]
\bibitem{AlickiFannes2013}
R.~Alicki and M.~Fannes,\
\emph{Entanglement boost for extractable work from ensembles of quantum batteries},\
%%title of arxiv version:\ Extractable work from ensembles of quantum batteries. Entanglement helps
%%DOI:\ 10.1103/PhysRevE.87.042123
\href{http://dx.doi.org/10.1103/PhysRevE.87.042123}{Phys.\ Rev.\ E\ \textbf{87}, 042123 (2013)}\
[\href{http://arxiv.org/abs/1211.1209}{arXiv:1211.1209}].
%%IOP Style
%%Alicki~R and Fannes~M\ 2013\ \textit{Phys.\ Rev.}\ E\ \href{http://dx.doi.org/10.1103/PhysRevE.87.042123}{\textbf{87} 042123}\
%%\textit{Preprint}\ arXiv:\href{http://arxiv.org/abs/1211.1209}{1211.1209} [quant-ph]

%%[8]
\bibitem{HovhannisyanPerarnauLlobetHuberAcin2013}
K.~V.~Hovhannisyan, M.~Perarnau-Llobet, M.~Huber, and A.~Ac$\acute{\i}$n,\
\emph{Entanglement Generation is Not Necessary for Optimal Work Extraction},\
%%DOI:\ 10.1103/PhysRevLett.111.240401
\href{http://dx.doi.org/10.1103/PhysRevLett.111.240401}{Phys.\ Rev.\ Lett.\ \textbf{111},}
\href{http://dx.doi.org/10.1103/PhysRevLett.111.240401}{240401 (2013)}\
[\href{http://arxiv.org/abs/1303.4686}{arXiv:1303.4686}].
%%IOP Style
%%Hovhannisyan~K~V, Perarnau-Llobet~M, Huber~M and Ac$\acute{\i}$n~A\ 2013\ \textit{Phys.\ Rev.\ Lett.}\
%%\href{http://dx.doi.org/10.1103/PhysRevLett.111.240401}{\textbf{111} 240401}\
%%\textit{Preprint}\ arXiv:\href{http://arxiv.org/abs/1303.4686}{1303.4686} [quant-ph]

%%[9]
\bibitem{PerarnauLlobetHovhannisyanHuberSkrzypczykTuraAcin2015}
M.~Perarnau-Llobet, K.~V.~Hovhannisyan, M.~Huber, P.~Skrzypczyk, J.~Tura, and A.~Ac$\acute{\i}$n,\
\emph{Most energetic passive states},\
%%DOI:\ 10.1103/PhysRevE.92.042147
\href{http://dx.doi.org/10.1103/PhysRevE.92.042147}{Phys.\ Rev.\ E\ \textbf{92}, 042147 (2015)}\
[\href{http://arxiv.org/abs/1502.07311}{arXiv:1502.07311}].
%%IOP Style
%%Perarnau-Llobet~M, Hovhannisyan~K~V, Huber~M, Skrzypczyk~P, Tura~J and Ac$\acute{\i}$n~A\ 2015\
%%\textit{Phys.\ Rev.}\ E\ \href{http://dx.doi.org/10.1103/PhysRevE.92.042147}{\textbf{92} 042147}\
%%\textit{Preprint}\ arXiv:\href{http://arxiv.org/abs/1502.07311}{1502.07311} [quant-ph]

%%[10]
\bibitem{AllahverdyanBalianNieuwenhuizen2004}
A.~E.~Allahverdyan, R.~Balian and Th.~M.~Nieuwenhuizen,\
\emph{Maximal work extraction from finite quantum systems},\
%%ODI:\ 10.1209/epl/i2004-10101-2
\href{http://dx.doi.org/10.1209/epl/i2004-10101-2}{Europhys.\ Lett.\ \textbf{67}, 565 (2004)}\
[\href{http://arxiv.org/abs/cond-mat/0401574}{arXiv:cond-}
\href{http://arxiv.org/abs/cond-mat/0401574}{mat/0401574}].
%%IOP Style
%%Allahverdyan~A~E, Balian~R and Nieuwenhuizen~Th~M\ 2004\ \textit{Europhys.\ Lett.}\ \href{http://dx.doi.org/10.1209/epl/i2004-10101-2}{\textbf{67} 565}\
%%\textit{Preprint}\ arXiv:cond-mat/\href{http://arxiv.org/abs/cond-mat/0401574}{0401574}

%%[11]
\bibitem{Aberg2013}
J.~{\AA}berg,\
\emph{Truly work-like work extraction via a single-shot analysis},\
%%title of arXiv version: \emph{Truly work-like work extraction},\
%%DOI:\ 10.1038/ncomms2712
\href{http://dx.doi.org/10.1038/ncomms2712}{Nat.\ Commun.\ \textbf{4}, 1925 (2013)}\
[\href{http://arxiv.org/abs/1110.6121}{arXiv:1110.6121}].
%%IOP Style
%%{\AA}berg~J\ 2013\ \href{Nat.\ Commun.}\ \href{http://dx.doi.org/10.1038/ncomms2712}{\textbf{4} 1925}\
%%\textit{Preprint}\ arXiv:\href{http://arxiv.org/abs/1110.6121}{1110.6121} [quant-ph]

%%[12]
\bibitem{SkrzypczykShortPopescu2014}
P.~Skrzypczyk, A.~J.~Short, and S.~Popescu,\
\emph{Work extraction and thermodynamics for individual quantum systems},\
%%title of arXiv v1: \emph{Thermodynamics for individual quantum systems},\
%%DOI:\ 10.1038/ncomms5185
%%e-print \href{http://arxiv.org/abs/1307.1558}{arXiv:1307.1558} [quant-ph] (2013).
\href{http://dx.doi.org/10.1038/ncomms5185}{Nat.\ Commun.\ \textbf{5}, 4185 (2014)}\
[\href{http://arxiv.org/abs/1307.1558}{arXiv:1307.1558}].
%%IOP Style
%%Skrzypczyk~P, Short~A~J and Popescu~S\ 2014\ \textit{Nat.\ Commun.}\ \href{http://dx.doi.org/10.1038/ncomms5185}{\textbf{5} 4185}\
%%\textit{Preprint}\ arXiv:\href{http://arxiv.org/abs/1307.1558}{1307.1558} [quant-ph]

%%[13]
\bibitem{SkrzypczykSilvaBrunner2015}
P.~Skrzypczyk, R.~Silva, and N.~Brunner,\
\emph{Passivity, complete passivity, and virtual temperatures},\
%%title of arxiv version: A short note on passivity, complete passivity and virtual temperatures
%%DOI:\ 10.1103/PhysRevE.91.052133
\href{http://dx.doi.org/10.1103/PhysRevE.91.052133}{Phys.\ Rev.\ E}\
\href{http://dx.doi.org/10.1103/PhysRevE.91.052133}{\textbf{91}, 052133 (2015)}\
[\href{http://arxiv.org/abs/1412.5485}{arXiv:1412.5485}].
%%IOP Style
%%Skrzypczyk~P, Silva~R and Brunner~N\ 2015\ \textit{Phys.\ Rev.}\ E\ \href{http://dx.doi.org/10.1103/PhysRevE.91.052133}{\textbf{91} 052133}\
%%\textit{Preprint}\ arXiv:\href{http://arxiv.org/abs/1412.5485}{1412.5485} [quant-ph]

%%[14]
\bibitem{GemmerAnders2015}
J.~Gemmer and J.~Anders,\
\emph{From single-shot towards general work extraction in a quantum thermodynamic framework},\
%%DOI:\ 10.1088/1367-2630/17/8/085006
\href{http://dx.doi.org/10.1088/1367-2630/17/8/085006}{New\ J.\ Phys.\ \textbf{17}, 085006 (2015)}\
[\href{https://arxiv.org/abs/1504.05061}{arXiv:1504.05061}].
%%IOP Style
%%Gemmer~J and Anders~J\ 2015\ \textit{New\ J.\ Phys.}\ \href{http://dx.doi.org/10.1088/1367-2630/17/8/085006}{\textbf{17} 085006}\
%%\textit{Preprint}\ arXiv:\href{https://arxiv.org/abs/1504.05061}{1504.05061} [quant-ph]

%%[15]
\bibitem{KammerlanderAnders2016}
P.~Kammerlander and J.~Anders,\
\emph{Coherence and measurement in quantum thermodynamics},\
%%DOI:\ 10.1038/srep22174
\href{http://dx.doi.org/10.1038/srep22174}{Sci.\ Rep.\ \textbf{6}, 22174}
\href{http://dx.doi.org/10.1038/srep22174}{(2016)}\
[\href{https://arxiv.org/abs/1502.02673}{arXiv:1502.02673}].
%%IOP Style
%%Kammerlander~P and Anders~J\ 2016\ \textit{Sci.\ Rep.}\ \href{http://dx.doi.org/10.1038/srep22174}{\textbf{6} 22174}\
%%\textit{Preprint}\ arXiv:\href{https://arxiv.org/abs/1502.02673}{1502.02673} [quant-ph]

%%[16]
\bibitem{KorzekwaLostaglioOppenheimJennings2016}
K.~Korzekwa, M.~Lostaglio, J.~Oppenheim, and D.~Jennings,\
\emph{The extraction of work from quantum coherence},\
%%DOI:\ 10.1088/1367-2630/18/2/023045
\href{http://dx.doi.org/10.1088/1367-2630/18/2/023045}{New\ J.\ Phys.\ \textbf{18}, 023045 (2016)}\
[\href{http://arxiv.org/abs/1506.07875}{arXiv:1506.07875}].
%%IOP Style
%%Korzekwa~K, Lostaglio~M, Oppenheim~J and Jennings~D\ 2016\ \textit{New\ J.\ Phys.}\ \href{http://dx.doi.org/10.1088/1367-2630/18/2/023045}{\textbf{18} 023045}\
%%\textit{Preprint}\ arXiv:\href{http://arxiv.org/abs/1506.07875}{1506.07875} [quant-ph]

%%[17]
\bibitem{PerarnauLlobetHovhannisyanHuberSkrzypczykBrunnerAcin2015}
M.~Perarnau-Llobet, K.~V.~Hovhannisyan, M.~Huber, P.~Skrzypczyk, N.~Brunner, and A.~Ac$\acute{\i}$n,\
\emph{Extractable work from correlations},\
%%title of arxiv v1: \emph{Extracting work from correlations},\
%%DOI:\ 10.1103/PhysRevX.5.041011
\href{http://dx.doi.org/10.1103/PhysRevX.5.041011}{Phys.\ Rev.\ X\ \textbf{5}, 041011 (2015)}\
[\href{http://arxiv.org/abs/1407.7765}{arXiv:1407.7765}].
%%IOP Style
%%Perarnau-Llobet~M, Hovhannisyan~K~V, Huber~M, Skrzypczyk~P, Brunner~N and Ac$\acute{\i}$n~A\ 2015\
%%\textit{Phys.\ Rev.}\ X\ \href{http://dx.doi.org/10.1103/PhysRevX.5.041011}{\textbf{5} 041011}\
%%\textit{Preprint}\ arXiv:\href{http://arxiv.org/abs/1407.7765}{1407.7765} [quant-ph]

%%[18]
\bibitem{HuberPerarnauHovhannisyanSkrzypczykKloecklBrunnerAcin2015}
M.~Huber, M.~Perarnau-Llobet, K.~V.~Hovhannisyan, P.~Skrzypczyk, C.~Kl{\"o}ckl, N.~Brunner, and A.~Ac$\acute{\i}$n,\
\emph{Thermodynamic cost of creating correlations},\
%%DOI:\ 10.1088/1367-2630/17/6/065008
\href{http://dx.doi.org/10.1088/1367-2630/17/6/065008}{New\ J.}\
\href{http://dx.doi.org/10.1088/1367-2630/17/6/065008}{Phys.\ \textbf{17}, 065008 (2015)}\
[\href{http://arxiv.org/abs/1404.2169}{arXiv:1404.2169}].
%%IOP Style
%%Huber~M, Perarnau-Llobet~M, Hovhannisyan~K~V, Skrzypczyk~P, Kl{\"o}ckl~C, Brunner~N and Ac$\acute{\i}$n~A\ 2015\
%%\textit{New\ J.\ Phys.}\ \href{http://dx.doi.org/10.1088/1367-2630/17/6/065008}{\textbf{17} 065008}\
%%\textit{Preprint}\ arXiv:\href{http://arxiv.org/abs/1404.2169}{1404.2169} [quant-ph]

%%[19]
\bibitem{BruschiPerarnauLlobetFriisHovhannisyanHuber2015}
D.~E.~Bruschi, M.~Perarnau-Llobet, N.~Friis, K.~V.~Hovhannisyan, and M.~Huber,\
\emph{The thermodynamics of creating correlations: Limitations and optimal protocols},\
%%DOI:\ 10.1103/PhysRevE.91.032118
%%e-print \href{http://arxiv.org/abs/1409.4647}{arXiv:1409.4647} [quant-ph] (2014).
\href{http://dx.doi.org/10.1103/PhysRevE.91.032118}{Phys.\ Rev.\ E\ \textbf{91}, 032118 (2015)}\
[\href{http://arxiv.org/abs/1409.4647}{arXiv:1409.4647}].
%%IOP Style
%%Bruschi~D~E, Perarnau-Llobet~M, Friis~N, Hovhannisyan~K~V and Huber~M\ 2015\ \textit{Phys.\ Rev.}\ E\
%%\href{http://dx.doi.org/10.1103/PhysRevE.91.032118}{\textbf{91} 032118}
%%\textit{Preprint}\ arXiv:\href{http://arxiv.org/abs/1409.4647}{1409.4647} [quant-ph]

%%[20]
\bibitem{FriisHuberPerarnauLlobet2016}
N.~Friis, M.~Huber, and M.~Perarnau-Llobet,\
\emph{Energetics of correlations in interacting systems},\
%%DOI:\ 10.1103/PhysRevE.93.042135
\href{http://dx.doi.org/10.1103/PhysRevE.93.042135}{Phys.\ Rev.\ E\ \textbf{93},}
\href{http://dx.doi.org/10.1103/PhysRevE.93.042135}{042135 (2016)}\
[\href{http://arxiv.org/abs/1511.08654}{arXiv:1511.08654}].
%%IOP Style
%%Friis~N, Huber~M and Perarnau-Llobet~M\ 2016\ \textit{Phys.\ Rev.}\ E\ \href{http://dx.doi.org/10.1103/PhysRevE.93.042135}{\textbf{93} 042135}\
%%\textit{Preprint}\ arXiv:\href{http://arxiv.org/abs/1511.08654}{1511.08654} [quant-ph]

%%[21]
\bibitem{MisraSinghBhattacharyaPati2016}
A.~Misra, U.~Singh, S.~Bhattacharya, and A.~K.~Pati,\
\emph{Energy Cost of Creating Quantum Coherence},\
%%DOI:\ 10.1103/PhysRevA.93.052335
\href{http://dx.doi.org/10.1103/PhysRevA.93.052335}{Phys.\ Rev.}\
\href{http://dx.doi.org/10.1103/PhysRevA.93.052335}{A\ \textbf{93}, 052335 (2016)}\
[\href{http://arxiv.org/abs/1602.08437}{arXiv:1602.08437}].
%%IOP Style
%%Misra~A, Singh~U, Bhattacharya~S and Pati~A~K\ 2016\ \textit{Phys.\ Rev.}\ A\ \href{http://dx.doi.org/10.1103/PhysRevA.93.052335}{\textbf{93} 052335}\
%%\textit{Preprint}\ arXiv:\href{http://arxiv.org/abs/1602.08437}{1602.08437} [quant-ph]

%%[22]
\bibitem{Olivares2012}
S.~Olivares,\
\emph{Quantum optics in the phase space - A tutorial on Gaussian states},\
%%DOI:\ 10.1140/epjst/e2012-01532-4
\href{http://dx.doi.org/10.1140/epjst/e2012-01532-4}{Eur.\ Phys.\ J.\ \textbf{203}, 3 %-24
(2012)}\ [\href{http://arxiv.org/abs/1111.0786}{arXiv:1111.0786}].
%%IOP Style
%%Olivares~S\ 2012\ \textit{Eur.\ Phys.\ J.}\ \href{http://dx.doi.org/10.1140/epjst/e2012-01532-4}{\textbf{203} 3\textendash24}\
%%\textit{Preprint}\ arXiv:\href{http://arxiv.org/abs/1111.0786}{1111.0786} [quant-ph]

%%[23]
\bibitem{Weedbrooketal2012}
C.~Weedbrook, S.~Pirandola, R.~Garc$\acute{\i}$a-Patr\'{o}n, N.~J.~Cerf, T.~C.~Ralph, J.~H.~Shapiro, and S.~Lloyd,\
\emph{Gaussian quantum information},\
%%DOI:\ 10.1103/RevModPhys.84.621
%%e-print \href{http://arxiv.org/abs/1110.3234}{arXiv:1110.3234} [quant-ph] (2012).
\href{http://dx.doi.org/10.1103/RevModPhys.84.621}{Rev.\ Mod.\ Phys.\ \textbf{84}, 621}
\href{http://dx.doi.org/10.1103/RevModPhys.84.621}{(2012)}\
[\href{http://arxiv.org/abs/1110.3234}{arXiv:1110.3234}].
%%IOP style
%%Weedbrook~C, Pirandola~S, Garc$\acute{\i}$a-Patr\'{o}n~R, Cerf~N~J, Ralph~T~C, Shapiro~J~H and Lloyd~S\ 2012\
%%\textit{Rev.\ Mod.\ Phys.}\ \href{http://dx.doi.org/10.1103/RevModPhys.84.621}{\textbf{84} 621\textendash69}\
%%\textit{Preprint}\ arXiv:\href{http://arxiv.org/abs/1110.3234}{1110.3234} [quant-ph]

%%[24]
\bibitem{LloydBraunstein1999}
S.~Lloyd and S.~L.~Braunstein,\
\emph{Quantum computation over continuous variables},\
%%DOI:\ 10.1103/PhysRevLett.82.1784
\href{http://dx.doi.org/10.1103/PhysRevLett.82.1784}{Phys.\ Rev.\ Lett.\ \textbf{82}, 1784}
\href{http://dx.doi.org/10.1103/PhysRevLett.82.1784}{(1999)}\
[\href{http://arxiv.org/abs/quant-ph/9810082}{arXiv:quant-ph/9810082}].
%%IOP Style
%%Lloyd~S and Braunstein~S~L\ 1999\ \textit{Phys.\ Rev.\ Lett.}\ \href{http://dx.doi.org/10.1103/PhysRevLett.82.1784}{\textbf{82} 1784\textendash7}\
%%\textit{Preprint}\ arXiv:quant-ph/\href{http://arxiv.org/abs/quant-ph/9810082}{9810082}

%%[25]
\bibitem{GelbwaserKlimovskyKurizki2015}
D.~Gelbwaser-Klimovsky and G.~Kurizki,\
\emph{Work extraction from heat-powered quantized optomechanical setups},\
%%DOI:\ 10.1038/srep07809
\href{http://dx.doi.org/10.1038/srep07809}{Sci.\ Rep.\ \textbf{5}, 7809 (2015)}\
[\href{https://arxiv.org/abs/1410.8561}{arXiv:1410.8561}].
%%IOP Style
%%Gelbwaser-Klimovsky~D and Kurizki~G\ 2015\ \textit{Sci.\ Rep.}\ \href{http://dx.doi.org/10.1038/srep07809}{\textbf{5} 7809}\
%%\textit{Preprint}\ arXiv:\href{https://arxiv.org/abs/1410.8561}{1410.8561} [quant-ph]

%%[26]
\bibitem{LevyDiosiKosloff2016}
A.~Levy, L.~Di{\'o}si, and R.~Kosloff,\
\emph{Quantum Flywheel},\
%%DOI:\ 10.1103/PhysRevA.93.052119
\href{https://doi.org/10.1103/PhysRevA.93.052119}{Phys.\ Rev.\ A\ \textbf{93}, 052119 (2016)}\
[\href{https://arxiv.org/abs/1602.04322}{arXiv:1602.04322}].
%%IOP Style
%%Levy~A, Di{\'o}si~L and Kosloff~R\ 2016\ \textit{Phys.\ Rev.}\ A\ \href{https://doi.org/10.1103/PhysRevA.93.052119}{\textbf{93} 052119}\
%%\textit{Preprint}\ arXiv:\href{https://arxiv.org/abs/1602.04322}{1602.04322} [quant-ph]

%%[27]
\bibitem{Williamson1936}
J.~Williamson,\
\emph{On the Algebraic Problem Concerning the Normal Forms of Linear Dynamical Systems},\
%%DOI:\ 10.2307/2371062
\href{http://dx.doi.org/10.2307/2371062}{Am.\ J.}\
\href{http://dx.doi.org/10.2307/2371062}{Math.\ \textbf{58}, 141 (1936)}.
%%IOP Style
%%Williamson~J\ 1936\ \textit{Am.\ J.\ Math.}\ \href{http://dx.doi.org/10.2307/2371062}{\textbf{58} 141\textendash63}

%%[28]
\bibitem{DuanGiedkeCiracZoller2000}
L.-M.~Duan, G.~Giedke, J.~I.~Cirac, and P.~Zoller,\
\emph{Inseparability Criterion for Continuous Variable Systems},\
%%DOI:\ 10.1103/PhysRevLett.84.2722
\href{http://dx.doi.org/10.1103/PhysRevLett.84.2722}{Phys.}\
\href{http://dx.doi.org/10.1103/PhysRevLett.84.2722}{Rev.\ Lett.\ \textbf{84}, 2722 (2000)}\
[\href{http://arxiv.org/abs/quant-ph/9908056}{arXiv:quant-ph/9908056}].
%%IOP Style
%%Duan~L-M, Giedke~G, Cirac~J~I and Zoller~P\ 2000\ \textit{Phys.\ Rev.\ Lett.}\ \href{http://dx.doi.org/10.1103/PhysRevLett.84.2722}{\textbf{84} 2722}\
%%\textit{Preprint}\ arXiv:quant-ph/\href{http://arxiv.org/abs/quant-ph/9908056}{9908056}

%%[29]
\bibitem{Simon2000}
R.~Simon,\
\emph{Peres-Horodecki Separability Criterion for Continuous Variable Systems},\
%%DOI:\ 10.1103/PhysRevLett.84.2726
\href{http://dx.doi.org/10.1103/PhysRevLett.84.2726}{Phys.\ Rev.\ Lett.\ \textbf{84}, 2726}
\href{http://dx.doi.org/10.1103/PhysRevLett.84.2726}{(2000)}\
[\href{http://arxiv.org/abs/quant-ph/9909044}{arXiv:quant-ph/9909044}].
%%IOP Style
%%Simon~R\ 2000\ \textit{Phys.\ Rev.\ Lett.}\ \href{http://dx.doi.org/10.1103/PhysRevLett.84.2726}{\textbf{84} 2726\textendash9}\
%%\textit{Preprint}\ arXiv:quant-ph/\href{http://arxiv.org/abs/quant-ph/9909044}{9909044}

%%[30]
\bibitem{RuppertUsenkoFilip2016}
L.~Ruppert, V.~C.~Usenko, and R.~Filip,\
\emph{Estimation of covariance matrix of macroscopic quantum states},\
%%DOI:\ 10.1103/PhysRevA.93.052114
\href{http://dx.doi.org/10.1103/PhysRevA.93.052114}{Phys.}\
\href{http://dx.doi.org/10.1103/PhysRevA.93.052114}{Rev.\ A\ \textbf{93}, 052114 (2016)}\
[\href{http://arxiv.org/abs/1511.06650}{arXiv:1511.06650}].
%%IOP Style
%%Ruppert~L, Usenko~V~C and Filip~R\ 2016\ \textit{Phys.\ Rev.}\ A\ \href{http://dx.doi.org/10.1103/PhysRevA.93.052114}{\textbf{93} 052114}\
%%\textit{Preprint}\ arXiv:\href{http://arxiv.org/abs/1511.06650}{1511.06650} [quant-ph]

%%[31]
\bibitem{Gamow1947HilbertHotel}
G.~Gamow,\ \emph{One Two Three... Infinity: Facts and Speculations of Science}\ (Viking Press, New York, 1947).

%%[32]
\bibitem{BruschiFriisFuentesWeinfurtner2013}
D.~E.~Bruschi, N.~Friis, I.~Fuentes, and S.~Weinfurtner,\
\emph{On the robustness of entanglement in analogue gravity systems},\
%%DOI:\ 10.1088/1367-2630/15/11/113016
%%e-print \href{http://arxiv.org/abs/arXiv:1305.3867}{arXiv:1305.3867} [quant-ph] (2013).
\href{http://dx.doi.org/10.1088/1367-2630/15/11/113016}{New\ J.\ Phys.\ \textbf{15}, 113016 (2013)}\
[\href{http://arxiv.org/abs/arXiv:1305.3867}{arXiv:1305.3867}].
%%IOP Style
%%Bruschi~D~E, Friis~N, Fuentes~I and Weinfurtner~S\ 2013\ \textit{New\ J.\ Phys.}\
%%\href{http://dx.doi.org/10.1088/1367-2630/15/11/113016}{\textbf{15} 113016}\
%%\textit{Preprint}\ arXiv:\href{http://arxiv.org/abs/arXiv:1305.3867}{1305.3867} [quant-ph]






\end{thebibliography}
\end{document}